\documentclass[english]{article}
\usepackage[T1]{fontenc}
\usepackage[latin9]{inputenc}
\usepackage[letterpaper]{geometry}
\usepackage{babel}
\usepackage{amsthm}
\usepackage{amsmath}
\usepackage{amssymb}
\usepackage{graphicx}
\usepackage{setspace}
\usepackage{esint}
\usepackage[unicode=true,pdfusetitle,
 bookmarks=true,bookmarksnumbered=false,bookmarksopen=false,
 breaklinks=false,pdfborder={0 0 1},backref=false,colorlinks=false]
 {hyperref}

\makeatletter

\numberwithin{equation}{section}
\newcommand{\lyxaddress}[1]{
\par {\raggedright #1
\vspace{1.4em}
\noindent\par}
}
\theoremstyle{plain}
\newtheorem{thm}{\protect\theoremname}[section]
  \theoremstyle{definition}
  \newtheorem{defn}[thm]{\protect\definitionname}
  \theoremstyle{remark}
  \newtheorem{rem}[thm]{\protect\remarkname}
  \theoremstyle{plain}
  \newtheorem{lem}[thm]{\protect\lemmaname}
  \theoremstyle{definition}
  \newtheorem{example}[thm]{\protect\examplename}
  \theoremstyle{plain}
  \newtheorem{prop}[thm]{\protect\propositionname}
  \theoremstyle{plain}
  \newtheorem{conjecture}[thm]{\protect\conjecturename}
  \theoremstyle{plain}
  \newtheorem{cor}[thm]{\protect\corollaryname}

\usepackage{ae,aecompl}  %Allows T1 encoding to look "normal"
\usepackage[]{hyperref}
\usepackage{url}
\usepackage{cancel} %enables strikethroughs in mathmode

\usepackage{cite}%Note:  This is not compatible with Natbib.  Natbib implements the functionality with its sort argument.

\numberwithin{equation}{section}
\date{}

\DeclareMathOperator*{\esssup}{ess\,sup} %Req amsmath package
\DeclareMathOperator*{\essinf}{ess\,inf} %Req amsmath package

\makeatother

  \providecommand{\conjecturename}{Conjecture}
  \providecommand{\corollaryname}{Corollary}
  \providecommand{\definitionname}{Definition}
  \providecommand{\examplename}{Example}
  \providecommand{\lemmaname}{Lemma}
  \providecommand{\propositionname}{Proposition}
  \providecommand{\remarkname}{Remark}
\providecommand{\theoremname}{Theorem}

\begin{document}

\title{Generalizations of Functionally Generated Portfolios with Applications
to Statistical Arbitrage

\global\long\global\long\global\long\def\norm#1{\left\Vert #1\right\Vert }

\global\long\global\long\global\long\def\abs#1{\left\vert #1\right\vert }

\global\long\global\long\global\long\def\set#1{\left\{  #1\right\}  }

\global\long\global\long\global\long\def\eps{\varepsilon}

\global\long\global\long\global\long\def\zero{\odot}

\global\long\global\long\global\long\def\fa{\mathfrak{a}}

\global\long\global\long\global\long\def\cA{\mathcal{A}}

\global\long\global\long\global\long\def\cB{\mathcal{B}}

\global\long\global\long\global\long\def\fb{\mathfrak{b}}

\global\long\global\long\global\long\def\C{\mathbb{C}}

\global\long\global\long\global\long\def\bC{\mathbb{C}}

\global\long\global\long\global\long\def\cC{\mathcal{C}}

\global\long\global\long\global\long\def\fC{\mathfrak{C}}

\global\long\global\long\global\long\def\cD{\mathcal{D}}

\global\long\global\long\global\long\def\bD{\mathbb{D}}

\global\long\global\long\global\long\def\rd{\mathrm{d}}

\global\long\global\long\global\long\def\d{\mathrm{d}}

\global\long\global\long\global\long\def\cE{\mathcal{E}}

\global\long\global\long\global\long\def\cF{\mathcal{F}}

\global\long\global\long\global\long\def\bF{\mathbb{F}}

\global\long\global\long\global\long\def\cG{\mathcal{G}}

\global\long\global\long\global\long\def\fG{\mathfrak{G}}

\global\long\global\long\global\long\def\fg{\mathfrak{g}}

\global\long\global\long\global\long\def\bG{\mathbb{G}}

\global\long\global\long\global\long\def\bH{\mathbb{H}}

\global\long\global\long\global\long\def\cH{\mathcal{H}}

\global\long\global\long\global\long\def\cI{\mathcal{I}}

\global\long\global\long\global\long\def\cK{\mathcal{K}}

\global\long\global\long\global\long\def\bK{\mathbb{K}}

\global\long\global\long\global\long\def\bL{\mathbb{L}}

\global\long\global\long\global\long\def\cL{\mathcal{L}}

\global\long\global\long\global\long\def\fL{\mathfrak{L}}

\global\long\global\long\global\long\def\cM{\mathcal{M}}

\global\long\global\long\global\long\def\fm{\mathfrak{m}}

\global\long\global\long\global\long\def\N{\mathbb{N}}

\global\long\global\long\global\long\def\bN{\mathbb{N}}

\global\long\global\long\global\long\def\cN{\mathcal{N}}

\global\long\global\long\global\long\def\cO{{\normalcolor \mathcal{O}}}

\global\long\global\long\global\long\def\bP{\mathbb{P}}

\global\long\global\long\global\long\def\cP{\mathcal{P}}

\global\long\global\long\global\long\def\fP{\mathfrak{P}}

\global\long\global\long\global\long\def\fp{\mathfrak{p}}

\global\long\global\long\global\long\def\bQ{\mathbb{Q}}

\global\long\global\long\global\long\def\cQ{\mathcal{Q}}

\global\long\global\long\global\long\def\fq{\mathfrak{q}}

\global\long\global\long\global\long\def\fr{\mathfrak{r}}

\global\long\global\long\global\long\def\R{\mathbb{R}}

\global\long\global\long\global\long\def\bR{\mathbb{R}}

\global\long\global\long\global\long\def\cR{\mathcal{R}}

\global\long\global\long\global\long\def\fR{\mathfrak{R}}

\global\long\global\long\global\long\def\bS{\mathbb{S}}

\global\long\global\long\global\long\def\cS{\mathcal{S}}

\global\long\global\long\global\long\def\fs{\mathfrak{s}}

\global\long\global\long\global\long\def\cT{\mathcal{T}}

\global\long\global\long\global\long\def\bT{\mathbb{T}}

\global\long\global\long\global\long\def\sT{\mathscr{T}}

\global\long\global\long\global\long\def\ft{\mathfrak{t}}

\global\long\global\long\global\long\def\fT{\mathfrak{T}}

\global\long\global\long\global\long\def\bU{\mathbb{U}}

\global\long\global\long\global\long\def\cU{\mathcal{U}}

\global\long\global\long\global\long\def\fu{\mathfrak{u}}

\global\long\global\long\global\long\def\bV{\mathbb{V}}

\global\long\global\long\global\long\def\cV{\mathcal{V}}

\global\long\global\long\global\long\def\fv{\mathfrak{v}}

\global\long\global\long\global\long\def\cX{\mathcal{X}}

\global\long\global\long\global\long\def\fX{\mathfrak{X}}

\global\long\global\long\global\long\def\cY{\mathcal{Y}}

\global\long\global\long\global\long\def\fY{\mathfrak{Y}}

\global\long\global\long\global\long\def\fy{\mathfrak{y}}

\global\long\global\long\global\long\def\bZ{\mathbb{Z}}

\global\long\global\long\global\long\def\cZ{\mathcal{Z}}

\global\long\global\long\global\long\def\fZ{\mathfrak{Z}}

\global\long\global\long\global\long\def\I{\mathbf{1}}

\global\long\global\long\global\long\def\cemetery{\dagger}

\global\long\global\long\global\long\def\D#1#2{\frac{\partial#1}{\partial#2}}

\global\long\global\long\global\long\def\DD#1#2{\frac{\partial^{2}#1}{\partial#2^{2}}}

$\global\long\global\long\global\long\def\vec#1{\mbox{\boldmath\ensuremath{#1}}}
$

\global\long\global\long\global\long\def\wt#1{\widetilde{#1}}

\global\long\global\long\global\long\def\1{\mathbf{1}}

\global\long\global\long\global\long\def\uI{\mathbf{\hat{1}}}

\global\long\global\long\global\long\def\2{\mathbf{1}}

\global\long\global\long\global\long\def\asto{\xrightarrow{\text{a.s.}}}

\global\long\global\long\global\long\def\Lto{\xrightarrow{L^{1}}}

\global\long\global\long\global\long\def\Lpto{\xrightarrow{L^{p}}}

\global\long\global\long\global\long\def\asLto{\xrightarrow{L^{1}, \text{ a.s.}}}

\global\long\global\long\global\long\def\imply{\Rightarrow}

\global\long\global\long\global\long\def\nimply{\nRightarrow}

\global\long\global\long\global\long\def\limply{\Longrightarrow}

\global\long\global\long\global\long\def\leftexp#1#2{{{\vphantom{#2}}^{#1}{#2}}}

\global\long\global\long\global\long\def\var{\textrm{Var}}

\global\long\global\long\global\long\def\std{\textrm{Std}}

\global\long\global\long\global\long\def\bias{\textrm{Bias}}

\global\long\global\long\global\long\def\corr{\textrm{Corr}}

\global\long\global\long\global\long\def\cov{\textrm{Cov}}

\global\long\global\long\global\long\def\tr{\textrm{tr}}

\global\long\global\long\global\long\def\diag{\textrm{diag}}

\global\long\global\long\global\long\def\sgn{\textrm{sign}}

\global\long\global\long\global\long\def\csch{\textrm{csch}}

\global\long\global\long\global\long\def\sech{\textrm{sech}}

\global\long\global\long\global\long\def\bbmid{\big|}

\global\long\global\long\global\long\def\bmid{\textrm{\ensuremath{\Big|}}}

\global\long\global\long\global\long\def\Bmid{\bigg|}

\global\long\global\long\global\long\def\BBmid{\Bigg|}

}

\author{Winslow Strong%
\thanks{The author gratefully acknowledges financial support from the National
Centre of Competence in Research ``Financial Valuation and Risk Management''
(NCCR FINRISK), project D1 (Mathematical Methods in Financial Risk
Management), as well as from the ETH Foundation.%
}}

\maketitle

\lyxaddress{\noindent \begin{center}
ETH Zürich, Department of Mathematics \\
CH-8092 Zürich, Switzerland\\
winslow.strong@math.ethz.ch
\par\end{center}}
\begin{abstract}
\noindent The theory of functionally generated portfolios (FGPs) is
an aspect of the continuous-time, continuous-path Stochastic Portfolio
Theory of Robert Fernholz. FGPs have been formulated to yield a \emph{master
equation} - a description of their return relative to a passive (buy-and-hold)
benchmark portfolio serving as the numéraire. This description has
proven to be analytically very useful, as it is both pathwise and
free of stochastic integrals. Here we generalize the class of FGPs
in several ways: (1) the numéraire may be any strictly positive wealth
process, not necessarily the market portfolio or even a passive portfolio;
(2) generating functions may be stochastically dynamic, adjusting
to changing market conditions through an auxiliary continuous-path
stochastic argument of finite variation. These generalizations do
not forfeit the important tractability properties of the associated
master equation. We show how these generalizations can be usefully
applied to scenario analysis, statistical arbitrage, portfolio risk
immunization, and the theory of mirror portfolios.
\end{abstract}
\begin{singlespace}


\end{singlespace}

\section{Introduction and background}

Functionally generated portfolios (FGPs) were introduced by Robert
Fernholz in \cite{Art:Fernholz:PortGenFuncts:1995,Art:Fernholz:PortGenFunct:1999},
see also \cite{Book:Fernholz:SPT:2002,Art:Karatzas&Fernholz:SPTReview:2009}.
They have historically been constructed by selecting a deterministic
generating function that takes the market portfolio as its argument.
They are notable for admitting a description of their performance,
relative to a passive (buy-and-hold) numéraire, that is both pathwise
and free of stochastic integrals. This description is known as the
\emph{master equation}, and is a useful tool for portfolio analysis
and optimization.

In markets that are uniformly elliptic and diverse \cite{Art:FernKaratzKard:DiversityAndRelArb:2005},
and more generally those markets with sufficient intrinsic volatility
\cite{Art:Fernholz&Karatzas:RelArbVolStab:2005}, FGPs yield explicit
portfolios that are arbitrages relative to the market portfolio (although
see \cite{Art:WinslowFouque:RegDivArb:2010} for an alternative diverse
market model that is compatible with no-arbitrage). In more general
equity market models, FGPs are useful for exploiting certain statistical
regularities, such as the stability of the distribution of capital
over time \cite{Book:Fernholz:SPT:2002,Art:Karatzas&Fernholz:SPTReview:2009},
and the non-constancy of the rate of variance of log-prices as a function
of sampling interval \cite{Art:Fernholz&Maguire:Stat_Arb}. These
FGP-derived portfolios are best described as \emph{statistical arbitrages
}\cite{Art:Fernholz&Maguire:Stat_Arb}, since they exploit the aforementioned
statistical regularities in the data to achieve favorable risk-return
profiles. One of the main attractions of the techniques presented
in this paper will undoubtedly be towards characterizing and optimizing
such statistical arbitrage portfolios.

This paper is organized as follows: Section \ref{Sec:Setting_Defs}
defines a market model typical of those used in Stochastic Portfolio
Theory. Section \ref{Sect:Generalizations} extends the class of FGPs
from its historical definition by allowing an \emph{arbitrary }wealth
process to serve as numéraire, rather than restricting it to be the
market portfolio or a more general passive portfolio. Generating functions
are also extended to accommodate continuous-path \emph{auxiliary stochastic
arguments }of finite variation\emph{. }Section \ref{Sub:Scenario}
highlights the usefulness of FGPs for \emph{scenario analysis}.\emph{
}Sections \ref{Sub:Trans_Equi} and \ref{Sub:Gauge} provide some
characterization of equivalence classes of FGPs in general, and in
the case of passive numéraires, respectively. Section \ref{Sec-Stat_Arb}
explores two approaches of applying FGPs to statistical arbitrage:
extending the original idea from \cite{Art:Fernholz&Maguire:Stat_Arb}
and using a new construction based on quadratic generating functions.
Section \ref{Sect-Immunization} presents a method of immunizing a
given FGP from certain market risks, while keeping it in the family
of FGPs. Section \ref{Sec-Mirror_Ports} extends the notion of mirror
portfolios introduced in \cite{Art:FernKaratzKard:DiversityAndRelArb:2005}
and analyzes their asymptotic behavior. Section \ref{Sect:Conclusion}
summarizes the results, and poses some remaining challenges to tackle
for the theory of FGPs.

\section{\noindent \label{Sec:Setting_Defs}Setting and definitions}

\noindent The market consists of $n$ processes with prices $X_{t}=(X_{1,t},\ldots,X_{n,t})^{\prime}$,
one of which may be a money market account. $X$ lives on a filtered
probability space $(\Omega,\cF,\mathbb{F}:=\{\cF_{t}\}_{t\ge0},\bP)$,
which supports a $d$-dimensional Brownian motion $W_{t}$, where
$d\ge n$. All processes introduced are assumed to be progressive
with respect to $\bF$, which may be strictly bigger than the Brownian
filtration.

\noindent Most of the analysis herein will take place on log prices
$L_{t}:=\log X_{t}$. The dynamics of $L$ are given by
\begin{align*}
dL_{t} & =\gamma_{t}dt+\sigma_{t}dW_{t},
\end{align*}
where $\gamma$ and $\sigma$ are $\bF$-progressive and satisfy
\begin{align*}
\sum_{i=1}^{n}\int_{0}^{t}\left(|\gamma_{i,s}|+a_{ii,s}\right)ds & <\infty,\quad\forall t\ge0,\\
\mbox{where }a_{ij,t}: & =[\sigma_{t}\sigma_{t}^{\prime}]_{ij}=\frac{d}{dt}\left\langle \log X_{i},\log X_{j}\right\rangle _{t},\quad1\le i,j\le n,
\end{align*}
and $\sigma_{t}$ takes values in $\R^{n\times d}$, $d\ge n$, $\forall t\ge0$.
Throughout, all equalities hold merely almost surely. The notation
$\I:=(1,\ldots,1)$ is used, where the dimensionality should be clear
from the context.
\begin{defn}
\noindent A \emph{portfolio $\pi$ on} $X$ is an $\bF$-progressively
measurable $\R^{n}$-valued process satisfying
\begin{align}
\int_{0}^{t}\left(\left|b_{\pi,s}\right|+\pi_{s}^{\prime}a_{s}\pi_{s}\right)ds & <\infty,\quad\forall t\ge0,\label{Eq:Integrability_Cond_to_be_a_Port}\\
\mbox{and }\sum_{i=1}^{n}\pi_{i,t} & =1,\quad\forall t\ge0.\nonumber
\end{align}
The \emph{wealth process $V^{v,\pi}$ }arising from investment according
to $\pi$ is given by
\begin{align*}
d\log V_{t}^{v,\pi} & =\gamma_{\pi,t}dt+\sigma_{\pi,t}dW_{t},\\
V_{0}^{v,\pi} & =v\in(0,\infty),\\
\mbox{where }\sigma_{\pi}: & =\pi^{\prime}\sigma,\\
\gamma_{\pi,t}: & =\pi_{t}^{\prime}\gamma_{t}+\gamma_{\pi,t}^{*},\\
\mbox{and }\gamma_{\pi,t}^{*}: & =\frac{1}{2}\left(\sum_{i=1}^{n}\pi_{i,t}a_{ii,t}-\pi_{t}^{\prime}a_{t}\pi_{t}\right).
\end{align*}

\end{defn}
\noindent The process $\gamma_{\pi}^{*}$ is called the \emph{excess
growth rate}, and plays an important role in Stochastic Portfolio
Theory \cite{Book:Fernholz:SPT:2002,Art:Karatzas&Fernholz:SPTReview:2009}.
To ease notation, we will use $V_{t}^{\pi}:=V_{t}^{1,\pi}$ and often
omit the subscript ``$t$'' when referring to processes.
\begin{rem}
\label{Rem:Synthetic}There is no need for $X$ to be restricted to
be the directly tradeable assets of a market. Some components may
also be wealth processes of portfolios on the tradeable assets, e.g.
$X_{i}=V^{\nu}$, where $\nu$ is a portfolio on the $m<n$ tradeable
assets. A weighting of $\pi_{i}$ in $X_{i}$ is equivalent to (additional)
weights of $\pi_{i}(\nu_{1},\ldots,\nu_{m})$ in the tradeable assets
(beyond what $\pi$ already explicitly specifies for those assets).
This flexibility may seem needlessly confusing, but it is useful when
portfolios are constructed with consideration to certain market segments
- e.g. value, growth, large/small cap, sectors, countries, etc. It's
also used in Section \ref{Sec-Stat_Arb} below. Distinguishing between
the directly tradeable assets and portfolios assembled on them can
be important if and when transaction costs and liquidity constraints
are taken into consideration, as turnover and leverage may be vastly
different. Those issues are beyond the scope of this paper.
\end{rem}
\noindent The following notation will prove useful:
\begin{align}
X^{\rho}: & =\frac{X}{V^{\rho}},\nonumber \\
L^{\rho}: & =\log X^{\rho},\nonumber \\
a_{ij}^{\rho}: & =[\sigma\sigma^{\prime}]_{ij}=\frac{d}{dt}\left\langle L_{i}^{\rho},L_{j}^{\rho}\right\rangle ,\nonumber \\
 & =a_{ij}-[a\rho]_{j}-[a\rho]_{i}+a_{\rho\rho},\label{Eq:Rel_Covar_in_terms_of_Covar}\\
a_{\pi\pi}: & =\sigma_{\pi}\sigma_{\pi}^{\prime}=\pi^{\prime}a\pi=\frac{d}{dt}\left\langle \log V^{\pi},\log V^{\pi}\right\rangle ,\nonumber \\
a_{\pi\pi}^{\rho}: & =\pi^{\prime}a^{\rho}\pi=\frac{d}{dt}\left\langle \log\left(\frac{V^{\pi}}{V^{\rho}}\right),\log\left(\frac{V^{\pi}}{V^{\rho}}\right)\right\rangle =a_{\rho\rho}^{\pi}.\nonumber
\end{align}
The \emph{numéraire invariance property }of $\gamma_{\pi}^{*}$ holds
for arbitrary portfolios $\pi$ and $\rho$:
\begin{align}
\gamma_{\pi}^{*} & =\frac{1}{2}\left(\sum_{i=1}^{n}\pi_{i}a_{ii}^{\rho}-a_{\pi\pi}^{\rho}\right),\label{Eq:Num_Invar_Prop_g_Star}
\end{align}

((3.5) of \cite{Art:Karatzas&Fernholz:SPTReview:2009}), and in particular,
since $a_{\rho\rho}^{\rho}=0$ by (\ref{Eq:Rel_Covar_in_terms_of_Covar}),
\begin{align}
\gamma_{\rho}^{*} & =\frac{1}{2}\sum_{i=1}^{n}\rho_{i}a_{ii}^{\rho}.\label{Eq:g_Star_Intr_Vol}
\end{align}
When $\rho$ is exclusively invested in the money market, then discounted
quantities will be denoted by $\hat{X}:=X^{\rho}$, $\hat{V}^{\pi}:=V^{\pi}/X^{\rho}$,
etc. Note also that in this case $a^{\rho}=a$.

\section{\noindent \label{Sect:Generalizations}Generalizations of functionally
generated portfolios}

Functionally generated portfolios were first introduced in \cite{Art:Fernholz:PortGenFuncts:1995,Art:Fernholz:PortGenFunct:1999},
see also \cite{Book:Fernholz:SPT:2002,Art:Karatzas&Fernholz:SPTReview:2009},
and the recent extension \cite{Art:Pal_Wong_Energy_Entropy_Arbitrage_2013}.
There are two generalizations presented here:
\begin{enumerate}
\item The numéraire in the master equation may be arbitrary. Previously,
it had been taken to be the market portfolio or some other passive
(buy-and-hold) portfolio.
\item Generating functions may take stochastic arguments, which here we
limit to finite-variation processes.
\end{enumerate}

\subsection{\noindent Stochastic generating functions and arbitrary numéraires}

\noindent It is natural to adjust a portfolio based on changing market
conditions. However, FGPs adjust their weights only as a deterministic
function of the underlying discounted price process $X^{\rho}$, which
doesn't allow for much flexibility. Ideally, one would like to be
able to modify the generating function stochastically while preserving
a useful pathwise description of relative return that is free from
stochastic integrals. As a step in this direction, time-dependent
generating functions have already been introduced in \cite{Book:Fernholz:SPT:2002}.
In this section we extend that idea to allow a dependence on auxiliary
stochastic processes of finite variation.

With respect to the historical work on portfolio generating functions,
we formulate them here in the log sense with logarithmic argument.
Specifically, our generating function $H$ is related to the previous
notion of generating function $G$ by $H(y)=\log G(e^{y})$. This
makes the analysis cleaner for our purposes.

\noindent For $H:\R^{n}\times\R^{k}\to\R$, let $\nabla_{l}$ be the
gradient with respect to the first ($n$-dim) argument of $H$, $\nabla_{f}$
be the gradient with respect to the second ($k$-dim) argument, and
$D_{l_{i}l_{j}}^{2}$, $D_{i,j}^{2}$ be the second-order differential
operator with respect to components $i$ and $j$ of the first argument
of $H$, and generally, respectively.
\begin{thm}
\noindent \label{Thm:Master_Eq_Stoch_GF}Let $H\in C^{2,1}(\R^{n}\times\R^{k},\R)$
and let $F$ be an $\R^{k}$-valued, $\bF$-progressive, continuous-path
process of finite variation. Then the portfolio
\begin{align}
\pi & =\lambda\rho+\nabla_{l}H(L^{\rho},F),\qquad\mbox{where }\lambda=1-\I^{\prime}\nabla_{l}H(L^{\rho},F),\label{Eq:FGP_Stoch_Def}
\end{align}
satisfies the following master equation:
\begin{align}
\log\left(\frac{V_{T}^{\pi}}{V_{T}^{\rho}}\right) & =H(L_{T}^{\rho},F_{T})-H(L_{0}^{\rho},F_{0})-\int_{0}^{T}\left[\nabla_{f}H(L_{t}^{\rho},F_{t})\right]^{\prime}dF_{t}+\int_{0}^{T}h_{t}dt,\label{Eq:Master_Stoch}\\
\mbox{where }h & =\gamma_{\pi}^{*}-\lambda\gamma_{\rho}^{*}-\frac{1}{2}\sum_{i,j=1}^{n}D_{l_{i}l_{j}}^{2}H\left(L^{\rho},F\right)a_{ij}^{\rho}.\nonumber
\end{align}
When the argument $F$ is not present (or constant), then it may be
suppressed. Hence
\begin{align}
\pi & =\lambda\rho+\nabla H(L^{\rho}),\qquad\mbox{where }\lambda=1-\I^{\prime}\nabla H(L^{\rho}),\label{Eq:FGP_def}
\end{align}
satisfies the following master equation:
\begin{align}
\log\left(\frac{V_{T}^{\pi}}{V_{T}^{\rho}}\right) & =H(L_{T}^{\rho})-H(L_{0}^{\rho})+\int_{0}^{T}h_{t}dt,\label{Eq:Master_General}\\
\mbox{where }h & =\gamma_{\pi}^{*}-\lambda\gamma_{\rho}^{*}-\frac{1}{2}\sum_{i,j=1}^{n}D_{ij}^{2}H\left(L^{\rho}\right)a_{ij}^{\rho}.\label{Eq:vol_cap_def_std}
\end{align}
\end{thm}
\begin{rem}
\noindent \label{Rem:On_Proof_Of_Master}Except for the change to
the log representation, the derivation proceeds analogously to the
original master equation \cite[Theorem 3.1]{Art:Fernholz:PortGenFuncts:1995},
which can also be found in \cite{Art:Fernholz:PortGenFunct:1999,Art:Karatzas&Fernholz:SPTReview:2009}.
The intermediate equations in the earlier derivations are each generalizable
to our setting, shown here as Lemmas \ref{Lem:Master_Int_Result_1}
and \ref{Lem:Rel_Ret_Self-Weighted_Is_Zero}. In the special (original)
case where $X$ is the total capitalization (shares $\times$ price
per share), then normalizing by the initial values so that the \#shares
of each asset is $1$ and choosing $\rho$ to be the market portfolio
results in
\begin{align}
\rho & =X/\sum_{i=1}^{n}X_{i}=X^{\rho}.\label{Eq:Numeraire_as_Market_Port}
\end{align}
Inserting $L:=\log X$ and this $\rho$ into (\ref{Eq:Master_General})
recovers the original master equation. However, in the general setting
of this paper, $\rho$ is arbitrary, making $X^{\rho}$ and $\rho$
distinct.
\end{rem}
\noindent The following two lemmas will be used in the proof of Theorem
\ref{Thm:Master_Eq_Stoch_GF}.
\begin{lem}
\noindent \label{Lem:Master_Int_Result_1}For any two portfolios $\pi$
and $\rho$, the following hold
\begin{align}
d\log\left(\frac{V^{\pi}}{V^{\rho}}\right) & =\sum_{i=1}^{n}\pi_{i}dL_{i}^{\rho}+\gamma_{\pi}^{*}dt,\label{Eq:Log_Rel_Ret_Log_Form}\\
 & =\sum_{i=1}^{n}\pi_{i}\frac{dX_{i}^{\rho}}{X_{i}^{\rho}}-\frac{1}{2}a_{\pi\pi}^{\rho}dt.\label{Eq:Log_Rel_Ret_Ari_Form}
\end{align}
\end{lem}
\begin{proof}
\noindent To prove (\ref{Eq:Log_Rel_Ret_Log_Form}), by definition
\begin{align*}
dL_{i}^{\rho} & =dL_{i}-d\log V^{\rho},\\
 & =\gamma_{i}dt+\sigma_{i}dW_{t}-d\log V^{\rho}.
\end{align*}
Plugging this into the right-hand side of (\ref{Eq:Log_Rel_Ret_Log_Form}),
we get
\begin{align*}
\sum_{i=1}^{n}\pi_{i}dL_{i}^{\rho}+\gamma_{\pi}^{*}dt & =\sum_{i=1}^{n}\pi_{i}(\gamma_{i}dt+\sigma_{i}dW_{t})-d\log V^{\rho}+\gamma_{\pi}^{*}dt,\\
 & =d\log V^{\pi}-d\log V^{\rho}.
\end{align*}
To prove (\ref{Eq:Log_Rel_Ret_Ari_Form}), use
\begin{align*}
dL_{i}^{\rho} & =\frac{dX_{i}^{\rho}}{X_{i}^{\rho}}-\frac{1}{2}\frac{d\left\langle X_{i}^{\rho},X_{i}^{\rho}\right\rangle }{\left(X_{i}^{\rho}\right)^{2}},\\
 & =\frac{dX_{i}^{\rho}}{X_{i}^{\rho}}-\frac{1}{2}a_{ii}^{\rho}.
\end{align*}
Plugging this into (\ref{Eq:Log_Rel_Ret_Log_Form}) and expanding
$\gamma_{\pi}^{*}$ with the numéraire invariance property (\ref{Eq:Num_Invar_Prop_g_Star})
 yields
\begin{align*}
d\log\left(\frac{V^{\pi}}{V^{\rho}}\right) & =\sum_{i=1}^{n}\pi_{i}\left(\frac{dX_{i}^{\rho}}{X_{i}^{\rho}}-\frac{1}{2}a_{ii}^{\rho}dt\right)+\frac{1}{2}\left(\sum_{i=1}^{n}\pi_{i}a_{ii}^{\rho}-a_{\pi\pi}^{\rho}\right)dt,\\
 & =\sum_{i=1}^{n}\pi_{i}\frac{dX_{i}^{\rho}}{X_{i}^{\rho}}-\frac{1}{2}a_{\pi\pi}^{\rho}dt.\tag*{\qedhere}
\end{align*}
\end{proof}
\begin{lem}
\noindent \label{Lem:Rel_Ret_Self-Weighted_Is_Zero}For any portfolio
$\rho$ on $X$,
\begin{align*}
\sum_{i=1}^{n}\rho_{i}dL_{i}^{\rho} & =-\gamma_{\rho}^{*}dt,\\
\sum_{i=1}^{n}\rho_{i}\frac{dX_{i}^{\rho}}{X_{i}^{\rho}} & =0.
\end{align*}
\end{lem}
\begin{proof}
\noindent For the first, use (\ref{Eq:Log_Rel_Ret_Log_Form}) with
$\pi=\rho$, and for the second use (\ref{Eq:Log_Rel_Ret_Ari_Form})
with $\pi=\rho$ and $a_{\rho\rho}^{\rho}=0$ from (\ref{Eq:Rel_Covar_in_terms_of_Covar}).
\end{proof}
\noindent Now we prove Theorem \ref{Thm:Master_Eq_Stoch_GF}.
\begin{proof}[Proof of Theorem \ref{Thm:Master_Eq_Stoch_GF}]
\noindent Initially, consider the case where $F\equiv1$, and hence
the second argument to $H$ may be suppressed. First plug in (\ref{Eq:FGP_def})
for $\pi$ into (\ref{Eq:Log_Rel_Ret_Log_Form}), and then get (\ref{Eq:Master_Proof_1})
by applying Lemma \ref{Lem:Rel_Ret_Self-Weighted_Is_Zero}:
\begin{align}
d\log\left(\frac{V^{\pi}}{V^{\rho}}\right) & =\sum_{i=1}^{n}\left(\lambda\rho_{i}+D_{i}H(L^{\rho})\right)dL_{i}^{\rho}+\gamma_{\pi}^{*}dt,\nonumber \\
 & =\sum_{i=1}^{n}D_{i}H(L^{\rho})dL_{i}^{\rho}+\left(\gamma_{\pi}^{*}-\lambda\gamma_{\rho}^{*}\right)dt,\label{Eq:Master_Proof_1}
\end{align}
where $D_{i}$ is the first derivative operator with respect to the
$i$th component. Expanding $dH(L^{\rho})$ gives
\begin{align*}
dH(L^{\rho}) & =\sum_{i=1}^{n}D_{i}H(L^{\rho})dL_{i}^{\rho}+\frac{1}{2}\sum_{i,j=1}^{n}D_{ij}^{2}H(L^{\rho})d\left\langle L_{i}^{\rho},L_{j}^{\rho}\right\rangle ,\\
 & =\sum_{i=1}^{n}D_{i}H(L^{\rho})dL_{i}^{\rho}+\frac{1}{2}\sum_{i,j=1}^{n}a_{ij}^{\rho}D_{ij}^{2}H(L^{\rho})dt.
\end{align*}
Plugging this into (\ref{Eq:Master_Proof_1}) yields
\begin{align}
d\log\left(\frac{V^{\pi}}{V^{\rho}}\right) & =dH(L^{\rho})+\left(\gamma_{\pi}^{*}-\lambda\gamma_{\rho}^{*}-\frac{1}{2}\sum_{i,j=1}^{n}a_{ij}^{\rho}D_{ij}^{2}H(L^{\rho})\right)dt,\label{Eq:Master_wo_Stoch}
\end{align}
proving the case when $F\equiv1$. For a finite variation $F$ with
continuous paths, the Itô-Doeblin formula yields
\begin{align*}
dH(L^{\rho},F) & =\left(dL^{\rho}\right)^{\prime}\nabla_{l}H(L^{\rho},F)+\left(dF\right)^{\prime}\nabla_{f}H(L^{\rho},F)+\frac{1}{2}\sum_{i,j=1}^{n}dL_{i}^{\rho}dL_{j}^{\rho}D_{l_{i}l_{j}}^{2}H(L^{\rho},F),
\end{align*}
which when combined with (\ref{Eq:Master_wo_Stoch}) proves the theorem.
\end{proof}
\noindent While adding an auxiliary stochastic process $F$ causes
FGPs to lose some elegance and tractability (comparing (\ref{Eq:Master_General})
to (\ref{Eq:Master_Stoch})), the extra flexibility gained can be
useful in practice. For example, $F$ may be factors that inform portfolio
construction, such as those of Fama and French \cite{Art:Fama&French_Cross_Sect_Exp_Stock_Ret_1992,Art:Fama&French_Common_Risk_Factors_1993},
fundamental economic data such as bond yields or stock market diversity
\cite{Art:Fernholz:OnDivEqMark:1999,Book:Fernholz:SPT:2002,Art:FernKaratzKard:DiversityAndRelArb:2005,Art:Karatzas&Fernholz:SPTReview:2009},
or information extracted from Twitter feeds \cite{Art:Bollen_etal_Twitter_Mood_Predicts_Stock_Market_2011}.
\begin{rem}[Generalizations]
It's possible to remove the restrictions on $F$ - that it's finite
variation and has continuous paths - to derive a more general master
equation, but this would make the correction term $\int_{0}^{T}\left[\nabla_{f}H(L_{t}^{\rho},F_{t})\right]^{\prime}dF_{t}$
of (\ref{Eq:Master_Stoch}) more complex. A continuous $F$ of finite
variation is sufficient for the applications that follow, so we do
not pursue these extensions here.

If a portfolio $\nu$ satisfies (\ref{Eq:Master_General}) in place
of $\pi$, then $V^{\nu}$ must be indistinguishable from $V^{\pi}$.
Hence, any differences between $\nu$ and $\pi$ are not meaningful
in the context of wealth-creation. However, there are portfolios obeying
generalizations of the master equation for which $H\notin C^{2}$,
such as the class presented in Theorem 4.1 of \cite{Art:Pamen-FGPs-2011}.
\end{rem}
\noindent The following example looks at the strategy of switching
from an initial FGP to a subsequent one at a stopping time. The overall
portfolio is an example of a stochastic FGP.
\begin{example}[Stochastic switching between FGPs%
\footnote{The author wishes to thank Radka Pickova for suggesting this idea.%
}]
Let $H_{1}$ and $H_{2}$ be arbitrary generating functions and let
$H$ be
\begin{align*}
H(y,i): & =iH_{1}(y)+(1-i)H_{2}(y),\qquad y\in\R^{n},\; i\in\R.
\end{align*}
For an arbitrary stopping time $\tau,$ $H(\cdot,\I_{t\le\tau})$
is a \emph{stochastic generating function} (i.e. a function $H$ meeting
the requirements of Theorem \ref{Thm:Master_Eq_Stoch_GF} with auxiliary
$F$) which will generate an FGP that switches at $\tau$ from $\pi^{(1)}$
generated by $H_{1}$ to $\pi^{(2)}$ generated by $H_{2}$. This
type of portfolio was used by Banner and D. Fernholz in \cite{Art:Banner&DFernholz:ShortTermRelArbVolStab:2008}
for constructing arbitrages relative to the market portfolio at arbitrarily
short deterministic horizons $T>0$ in a class of market models including
volatility-stabilized markets \cite{Art:Fernholz&Karatzas:RelArbVolStab:2005}.
Those models have also been shown to admit functionally generated
relative arbitrage over sufficiently long time horizons \cite{Art:Fernholz&Karatzas:RelArbVolStab:2005}.
However, there is a horizon before which functionally generated relative
arbitrage is not possible, regardless of the choice of generating
function \cite{UNP:Whitman}. This example shows that relative arbitrages
exist on arbitrarily short horizons within the class of FGPs that
have \emph{stochastic }generating functions.
\end{example}

\subsection{\label{Sub:Scenario}Pathwise returns for scenario analysis}

One of the main analytical benefits of the master equation is that
it is free of stochastic integrals. When formulas for portfolio returns
contain stochastic integrals, then correct analysis may be counterintuitive,
and incorrect analysis may be intuitively appealing, as the following
example demonstrates.
\begin{example}
\label{Ex:Port_Close_Ret_Not_Close} Consider two portfolios on horizons
$0\le t\le T$: Let $\pi$ be a long-only constant-weight portfolio
so that $\pi_{t}=p$, $0\le t\le T$, and let $\tilde{\pi}$ be a
passive (buy-and-hold) portfolio starting from the same initial allocation
$\tilde{\pi}_{0}=p$. Consider the set $A_{\varepsilon}:=\left\{ \omega\in\Omega\mid\left\Vert \tilde{\pi}_{t}(\omega)-\pi_{t}(\omega)\right\Vert <\varepsilon,\;0\le t\le T\right\} $.
In some models for $X$, (e.g. geometric Brownian motion), $P(A_{\varepsilon})>0$,
$\forall\varepsilon>0$. Recalling that the usual description of the
wealth process of an arbitrary portfolio $\theta$ is
\begin{align}
\log V_{T}^{\theta} & =\int_{0}^{T}\gamma_{\theta,t}dt+\int_{0}^{T}\theta_{t}^{\prime}\sigma_{t}dW_{t},\label{Eq:Port_Return_Usual_Description}
\end{align}
then we may write
\begin{align*}
\left|\log V_{T}^{\pi}-\log V_{T}^{\tilde{\pi}}\right| & \le\int_{0}^{T}\left|\gamma_{\pi,t}-\gamma_{\tilde{\pi},t}\right|dt+\left|\int_{0}^{T}\left(\pi_{t}-\tilde{\pi}_{t}\right)^{\prime}\sigma_{t}dW_{t}\right|.
\end{align*}
If $\gamma$ and $\sigma$ are bounded, then it is tempting to make
the \emph{erroneous} conclusion that
\begin{align}
\lim_{\varepsilon\to0}\esssup_{\omega\in A_{\varepsilon}}\left|\log V_{T}^{\pi}(\omega)-\log V_{T}^{\tilde{\pi}}(\omega)\right| & =0.\label{Eq:Erroneous_Port_Return}
\end{align}
The \emph{erroneous} conclusion might be stated in words as:
\begin{quotation}
\emph{If two portfolios remain sufficiently close to each other, then
their returns must be close.}
\end{quotation}
The erroneous conclusion can be avoided by noting that $\pi$ and
$\tilde{\pi}$ are functionally generated by generating functions
$H(y)=\sum_{i=1}^{n}p_{i}y_{i}$, and $\tilde{H}(y)=\log\left(\sum_{i=1}^{n}p_{i}e^{y_{i}-l_{i}}\right)$,
respectively, where $l=L_{0}\in\R^{n}$. Comparing their wealth processes
via their master equations gives the pathwise equation
\begin{align}
\log V_{T}^{\pi}-\log V_{T}^{\tilde{\pi}} & =H(L_{T})-H(L_{0})-\left(\tilde{H}(L_{T})-\tilde{H}(L_{0})\right)+\int_{0}^{T}\gamma_{p,s}^{*}ds.\label{Eq:Rel_Ret_Discrepancy_CWP_Passive}
\end{align}
If the covariance $a$ is uniformly elliptic (there exists $u>0$
such that $y^{\prime}a_{t}y\ge u\left\Vert y\right\Vert ^{2},$ for
all $y\in\R^{n}$, $t\ge0$), and if $p_{i}>0$, $\forall i$, then
$\gamma_{p,t}^{*}\ge u\in(0,\infty)$, $\forall t\ge0$ \cite[Lemma 3.4]{Art:Karatzas&Fernholz:SPTReview:2009}.
For all $\xi>0$ there exists $\varepsilon>0$ such that $\left\Vert L_{T}-L_{0}\right\Vert <\xi$
on $A_{\varepsilon}$. Hence, (\ref{Eq:Rel_Ret_Discrepancy_CWP_Passive})
and the continuity of $H$ and $\tilde{H}$ imply that
\begin{align*}
\lim_{\varepsilon\to0}\essinf_{A_{\varepsilon}}\left\{ \log V_{T}^{\pi}-\log V_{T}^{\tilde{\pi}}\right\}  & =\lim_{\varepsilon\to0}\essinf_{A_{\varepsilon}}\left\{ \int_{0}^{T}\gamma_{p,s}^{*}ds\right\} \ge Tu,
\end{align*}
contradicting (\ref{Eq:Erroneous_Port_Return}). This correct conclusion
is simply obtainable from the pathwise representation of return given
by the master equation, but is difficult to arrive at from the traditional
representation of return given in (\ref{Eq:Port_Return_Usual_Description}).
The upshot is that
\begin{quotation}
\emph{Just because two portfolios remain arbitrarily close does not
imply that their returns are close.}
\end{quotation}
\end{example}
More generally, the master-equation description of relative return
(\ref{Eq:Master_General}) has advantages over the usual description
(\ref{Eq:Port_Return_Usual_Description}) for \emph{scenario analysis},
a technique currently popular in investment management (e.g. see \cite{Art:Mina&Xiao_Return_to_RiskMetrics}).
In the short term, an FGP's performance (particularly its potential
for loss) is largely attributable to the first term of (\ref{Eq:Master_General}),
which is entirely determined by the terminal values of the underlying
assets, which themselves are outputs of scenario analysis. Whereas
the last term involves the quadratic variation of the path, and is
usually easier to estimate with high precision. In contrast, knowledge
of the terminal values of the assets is difficult to use in the Itô-integral
formulation (\ref{Eq:Port_Return_Usual_Description}).

\subsection{\noindent \label{Sub:Trans_Equi}Translation equivariance and numéraire
invariance}

\noindent Generating functions are overspecified in the following
sense. Given a generating function $H$, each member of the equivalence
class of generating functions
\begin{align}
[H] & :=\{\kappa+H\mid\kappa\in\R\}\label{Eq:Equiv_Class_Linear_Scaling}
\end{align}
yields the same function $\nabla H$. Hence, given an arbitrary market
$X$ and numéraire $\rho$, any member of $[H]$ yields the same functionally
generated portfolio (\ref{Eq:FGP_def}).
\begin{defn}
\noindent $H:\R^{n}\to\R$ is \emph{translation equivariant }if
\begin{align*}
H(y+\I\kappa) & =\kappa+H(y),\quad\forall y\in\R^{n}.
\end{align*}

\end{defn}
\noindent When $H$ is translation equivariant, then $H$ and hence
its corresponding FGP depend only on the \emph{relative} rather than
\emph{absolute} price level. An example of a class of translation-equivariant
generating functions is the diversity-$p$ family (see \cite{Book:Fernholz:SPT:2002,Art:Karatzas&Fernholz:SPTReview:2009}):
\begin{align*}
H_{p}(y) & =\frac{1}{p}\log\left(\sum_{i=1}^{n}\exp\{py_{i}\}\right),\qquad y\in\R^{n}.
\end{align*}
 The following is an invariance property of the master equation
when $H$ is translation equivariant.
\begin{prop}
\noindent \label{Prop:Master_Eq_for_PH_GF}If $H\in C^{2}(\R^{n},\R)$
is translation-equivariant, then $\lambda=0$ in (\ref{Eq:FGP_def}),
and (\ref{Eq:Master_General}) exhibits no sensitivity to numéraire
choice. Specifically, for arbitrary portfolios $\rho$, $\phi$, $\eta$,
and $\nu$,
\begin{align}
\log\frac{V_{T}^{\pi}}{V_{T}^{\rho}} & =H(L_{T}^{\rho})-H(L_{0}^{\rho})+\int_{0}^{T}h_{t}dt,\label{Eq:Master_Trans_Equi}\\
\mbox{where }\pi & =\nabla\log H(L^{\phi}),\nonumber \\
\mbox{and }h & =\gamma_{\pi}-\frac{1}{2}\sum_{i,j=1}^{n}D_{ij}^{2}H(L^{\eta})a_{ij}^{\nu}.\label{Eq:g_For_PH_FGPs}
\end{align}
\end{prop}
\begin{proof}
\noindent Starting with (\ref{Eq:Master_General}), we show that when
$H$ is translation equivariant, then $\lambda$ of (\ref{Eq:FGP_def})
is identically $0$:
\begin{align*}
\frac{\partial}{\partial\kappa}H(y+\I\kappa) & =\frac{\partial}{\partial\kappa}\left(H(y)+\kappa\right)=1,\\
\frac{\partial}{\partial\kappa}H(y+\I\kappa) & =\sum_{i=1}^{n}\frac{\partial(y_{i}+\kappa)}{\partial\kappa}D_{i}H(y+\I\kappa)=\sum_{i=1}^{n}D_{i}H(y+\I\kappa)=1-\lambda.
\end{align*}
Next, we show that $a^{\rho}$ may be formally replaced with $a^{\nu}$
in (\ref{Eq:vol_cap_def_std}). We note that
\begin{align*}
\sum_{i=1}^{n}D_{ij}^{2}H(y) & =D_{j}\sum_{i=1}^{n}D_{i}H(y)=D_{j}\left(1\right)=0,\quad\forall y\in\R^{n}.
\end{align*}
Using this and the form (\ref{Eq:Rel_Covar_in_terms_of_Covar}) for
$a^{\rho}$ yields
\begin{align*}
\sum_{i,j=1}^{n}D_{ij}^{2}Ha_{ij}^{\rho} & =\sum_{i,j=1}^{n}\left(a_{ij}-[a\rho]_{j}-[a\rho]_{i}-a_{\rho\rho}\right)D_{ij}^{2}H,\\
 & =\sum_{i,j=1}^{n}D_{ij}^{2}Ha_{ij}-2\sum_{j}[a\rho]_{j}D_{j}\sum_{i=1}^{n}D_{i}H-a_{\rho\rho}\sum_{i,j=1}^{n}D_{ij}^{2}H,\\
 & =\sum_{i,j=1}^{n}D_{ij}^{2}Ha_{ij}.
\end{align*}
Reversing the steps shows that $a$ may be replaced with $a^{\nu}$
for arbitrary $\nu$.

\noindent It remains to show that $L^{\rho}$ may be replaced with
$L^{\eta}$ as the argument to $D^{2}H$:
\begin{align*}
D_{ij}^{2}H(y+\I\kappa) & =D_{ij}^{2}\left(H(y)+\kappa\right)=D_{ij}^{2}H(y),\quad\forall y\in\R^{n},\;\forall\kappa\in\R.
\end{align*}
Since $L^{\rho}=L-\log V^{\rho}$, this implies that
\begin{align*}
h & =\gamma_{\pi}^{*}-\frac{1}{2}\sum_{i,j=1}^{n}D_{ij}^{2}H(L^{\rho})a_{ij}^{\nu}=\gamma_{\pi}^{*}-\frac{1}{2}\sum_{i,j=1}^{n}D_{ij}^{2}H(L^{\eta})a_{ij}^{\nu}.\tag*{\qedhere}
\end{align*}
\end{proof}
\begin{rem}
An FGP can be thought of as a $\Delta$-hedge for its generating function,
as it eliminates stochastic integrals from $dH(L^{\rho})$. Under
this interpretation, $\lambda$ is the position taken in the numéraire
with the leftover money in the portfolio. The numéraire $\rho$ sets
a stochastic relative price level for $V^{\pi}$ and $L$ in the master
equation. When $H$ is translation equivariant, then the corresponding
FGP and wealth process have no sensitivity to the price level, as
can be seen by Proposition \ref{Prop:Master_Eq_for_PH_GF} and
\begin{align*}
H(L^{\rho}) & =H(L-\I\log V^{\rho})=H(L)-\log V^{\rho},
\end{align*}
which simplifies (\ref{Eq:Master_Trans_Equi}) to
\begin{align}
\log V_{T}^{\pi} & =H(L_{T})-H(L_{0})+\int_{0}^{T}h_{t}dt.\label{Eq:Master_Eq_Trans_Eq}
\end{align}
Generally, each choice of numéraire results in a unique master equation.
But when $H$ is translation equivariant, then there's no excess exposure
to the numéraire (beyond what's needed to $\Delta$-hedge $H$), making
the master equations arising from different numéraire choices trivial
translations of the same one equation (\ref{Eq:Master_Eq_Trans_Eq}).
\end{rem}

\subsection{\noindent \label{Sub:Gauge}Passive numéraires and gauge freedom}

The historical work on FGPs \cite{Art:Fernholz:PortGenFuncts:1995,Art:Fernholz:PortGenFunct:1999,Book:Fernholz:SPT:2002,Art:Karatzas&Fernholz:SPTReview:2009}
takes $X$ as the total capitalizations  and the numéraire $\rho$
as the market portfolio, leading to $\rho=X^{\rho}$ (see \ref{Eq:Numeraire_as_Market_Port}),
as in Remark \ref{Rem:On_Proof_Of_Master}. The important property
of the market portfolio that was exploited in those works was its
passivity on $X$ (see Definition \ref{Def:Passive}). In the traditional
case, the equality of $\rho$ and $X^{\rho}$ means that $\sum_{i=1}^{n}X_{i}^{\rho}=1$,
hence the generating function need not be defined on all of $\R^{n}$.
In this section we explore more generally to what extent a passive
numéraire allows a reduction of the domain of the generating function.
\begin{defn}
\noindent \label{Def:Passive}A portfolio $\rho$ is \emph{passive
}if there exists a constant $s\in\R^{n}$, called the \emph{shares},
such that
\begin{align*}
V^{\rho} & =s^{\prime}X\qquad\mbox{and}\qquad\rho_{i}=\frac{s_{i}X_{i}}{s^{\prime}X},\quad1\le i\le n.
\end{align*}

\end{defn}
\noindent Passive portfolios are untraded after the initial allocation,
so are unaffected by transaction costs and other liquidity concerns.
Generally, the $X_{i}$ are unbounded from above, so in order that
$V^{\rho}>0$ is guaranteed, we assume henceforth that any passive
portfolio is long-only. That is, that $s\in[0,\infty)^{n}\setminus\{0\}$.

\noindent When the numéraire $\rho$ is passive, then a generating
function need not be defined on all of $\R^{n}$, as $X^{\rho}$ will
be confined to a hyperplane. Let $s\in[0,\infty)^{n}\setminus\{0\}$
be the constant vector of shares such that $V^{\rho}=s^{\prime}X$.
Then,
\begin{align*}
s^{\prime}X^{\rho} & =\frac{s^{\prime}X}{V^{\rho}}=\sum_{i=1}^{n}\rho_{i}=1.
\end{align*}
Thus, $X^{\rho}$ is confined to a hyperplane of codimension $1$,
and it should be sufficient to define $H$ on
\begin{align*}
E_{s}^{n}: & =\{y\in\R^{n}\mid\sum_{i=1}^{n}s_{i}e^{y_{i}}=1\}.
\end{align*}
However, Theorem \ref{Thm:Master_Eq_Stoch_GF} uses the Cartesian
coordinate system, which is quite convenient, so we sacrifice some
generality and require generating functions to be defined on a neighborhood
in $\R^{n}$ containing $E_{s}^{n}$.
\begin{defn}
\noindent Let $s\in[0,\infty)^{n}\setminus\{0\}$ and let the passive
portfolio $\rho$ be given by $\rho_{i}=s_{i}X_{i}/s^{\prime}X$,
$1\le i\le n$. A $\rho$-\emph{generating function }is a function
$H\in\cC^{2}(U,\R)$, where $U$ is a neighborhood in $\R^{n}$ containing
$E_{s}^{n}$.
\end{defn}
\noindent When $\rho$ is the market portfolio and $X$ is the total
capitalization, then $s\propto\I$. In \cite[Proposition 3.1.14]{Book:Fernholz:SPT:2002},
where generating functions are specified as $G(x):=\exp\{H(\log x)\}$,
the following equivalence is demonstrated: Generating functions $H_{1}$
and $H_{2}$ generate the same portfolio if and only if $H_{1}-H_{2}$
is constant on $E_{\I}^{n}$. This generalizes to general passive
portfolios as follows.
\begin{prop}
\noindent \label{Prop:Gauge_Invar}Let \textup{$\rho$} be passive
with corresponding shares $s\in[0,\infty)^{n}\setminus\{0\}$. Let
$H_{1}$ and $H_{2}$ be two $\rho$-generating functions defined
on a neighborhood $U$ containing $E_{s}^{n}$. Then $H_{1}$ and
$H_{2}$ generate the same portfolio for any realization of $X$ if
and only if $H_{1}-H_{2}$ is constant on \textup{$E_{s}^{n}$.} \end{prop}
\begin{proof}
\noindent Let $\pi^{j}$ be the portfolio generated by $H_{j}$, $j\in\{1,2\}$.
The condition $\pi_{i}^{1}=\pi_{i}^{2}$, $1\le i\le n$ for all realizations
of $X$ is equivalent by (\ref{Eq:FGP_def}) to the following holding
$\forall y\in E_{s}^{n}$:
\begin{align}
\left(1-\sum_{j=1}^{n}D_{j}H_{1}(y)\right)\rho_{i}+D_{i}H_{1}(y) & =\left(1-\sum_{j=1}^{n}D_{j}H_{2}(y)\right)\rho_{i}+D_{i}H_{2}(y),\quad1\le i\le n,\nonumber \\
\Longleftrightarrow\frac{s_{i}e^{y_{i}}}{\sum_{k=1}^{n}s_{k}e^{y_{k}}}\sum_{j=1}^{n}D_{j}\left(H_{1}(y)-H_{2}(y)\right) & =D_{i}(H_{1}(y)-H_{2}(y)),\quad1\le i\le n,\nonumber \\
\imply(s_{1}e^{y_{1}},\ldots,s_{n}e^{y_{n}}) & \propto\nabla\left(H_{1}(y)-H_{2}(y)\right),\quad1\le i\le n,\label{Eq:Gauge_Invar_Proof_1}
\end{align}
Differentiating the equation determining the surface $E_{s}^{n}$
shows that (\ref{Eq:Gauge_Invar_Proof_1}) is equivalent to $\nabla(H_{1}-H_{2})$
being orthogonal to $E_{s}^{n}$, hence equivalent to $H_{1}-H_{2}$
being constant on $E_{s}^{n}$. Conversely, if $H_{1}-H_{2}$ is constant
on $E_{s}^{n}$, then (\ref{Eq:Gauge_Invar_Proof_1}) holds. Since
by definition $\rho\propto(s_{1}e^{y_{1}},\ldots,s_{n}e^{y_{n}})$,
then from (\ref{Eq:FGP_def}) $\pi^{1}-\pi^{2}\propto\rho$. But $\rho$
is long-only, and $\sum_{i=1}^{n}\left(\pi_{i}^{1}-\pi_{i}^{2}\right)=0$.
Thus $\pi^{1}=\pi^{2}$.
\end{proof}
\noindent The gauge freedom implied by Proposition \ref{Prop:Gauge_Invar},
specifically by (\ref{Eq:Gauge_Invar_Proof_1}), is that if $H$ generates
$\pi$, then
\begin{align*}
H_{f}(y): & =f\left(\sum_{i=1}^{n}s_{i}e^{y_{i}}\right)+H(y),\quad y\in U\supset E_{s}^{n},
\end{align*}
also generates $\pi$, for any $f\in\cC^{2}((0,\infty),\R)$. This
gauge freedom allows one to make any convenient choice for $f$ in
order to simplify calculations. The $\rho$-generating functions and
general generating functions have their associated respective equivalence
classes:
\begin{align*}
[H]_{\rho} & :=\{f(s^{\prime}e^{\cdot})+H(\cdot)\mid f\in\cC^{2}((0,\infty),\R)\},\quad\mbox{where }s_{i}:=\frac{\rho_{i,0}}{X_{i,0}},\;1\le i\le n,\\
{}[H] & :=\{H+\kappa\mid\kappa\in\R\}.
\end{align*}
Each member of a given $[H]_{\rho}$ or $[H]$ is equivalent for the
purposes of $\rho$-FGPs, or generally FGPs, respectively.

\section{\label{Sec-Stat_Arb}Statistical arbitrage}

\subsection{\label{Sub-Long-Short_Stat_Arb}Long-short statistical arbitrage
with FGPs}

The paper \cite{Art:Fernholz&Maguire:Stat_Arb} of R. Fernholz and
C. Maguire introduces an idea for a statistical arbitrage strategy
in markets where the realized rate of variance of log market prices
depends on the sampling interval. The general idea is to take a long
position in an FGP that is rebalanced over a time interval corresponding
to a high variance rate, hedged with a short position in an FGP generated
from the \emph{same} generating function, but rebalanced over a \emph{different
}time interval corresponding to a low variance rate. Statistical arbitrage
profits accrue from the different rates of \emph{variance-capture}
(the $h$ of (\ref{Eq:vol_cap_def_std})) - named so because these
terms are directly proportional to the variance rate. Because the
long and short FGPs have the same generating function, their corresponding
$H$ terms of (\ref{Eq:Master_General}) are identical, providing
an effective hedge for each other.

The data presented in \cite{Art:Fernholz&Maguire:Stat_Arb} indicate
that for $2005$, the variance rate was significantly higher at higher
sampling frequencies intradaily for large-cap US equities. The authors
looked at rebalancing the long component at $90$-second intervals
and rebalancing the short component once a day. These choices of rebalancing
intervals were ad hoc, not the output of an optimization problem.
In this section we develop general performance formulas for such long-short
statistical arbitrages, creating a framework for optimizing the selection
of generating function and rebalancing intervals.

We will show that the growth rate of the statistical arbitrage portfolio
always has the quadratic form
\begin{align}
\gamma_{\pi} & =A\kappa-B\kappa^{2},\qquad A,B>0,\label{Eq:Quad_Form_of_FGP_Stat_Arb}
\end{align}
where $\kappa$ is the leverage factor, that is, the weight invested
in the long portfolio. Hence, there is a level of leverage $\bar{\kappa}=A/B$
above which the portfolio tends to shrink in value rather than grow.
The leverage $\check{\kappa}=A/(2B)=\bar{\kappa}/2$ gives the maximal
growth rate of $\check{\gamma}_{\pi}=A^{2}/(4B)$.

To estimate the performance of the strategy described in \cite{Art:Fernholz&Maguire:Stat_Arb},
we use $X=(X_{1},X_{2},X_{3},\ldots,X_{3+m})$, where $X_{3}$ is
the money market, and $X_{3+1},\ldots,X_{3+m}$ are the risky assets
that are directly tradeable on the market (e.g. the equities). $X_{1}$
and $X_{2}$ are the values of long-only portfolios on $(X_{4},\ldots,X_{3+m})$.
The statistical arbitrage portfolio is an FGP specified on the submarket
$(X_{1},X_{2},X_{3})$ (for more detail, see Remark \ref{Rem:Synthetic}).
In \cite{Art:Fernholz&Maguire:Stat_Arb} constant-weight FGPs are
considered, where the overall portfolio $\pi$ has $\pi_{1}=-\pi_{2}=\kappa\in(0,\infty)$,
and $\pi_{3}=1$. This portfolio is functionally generated by the
generating function $H(y)=\kappa\left(y_{1}-y_{2}\right)$, so (\ref{Eq:Master_General})
yields
\begin{align*}
\log\hat{V}_{T}^{\pi} & =\kappa\left(\hat{L}_{1,T}-\hat{L}_{1,0}-\left[\hat{L}_{2,T}-\hat{L}_{2,0}\right]\right)+\int_{0}^{T}h_{s}ds.
\end{align*}
The statistical arbitrage construction uses $X_{1}$ and $X_{2}$
as discretely-traded approximations to \emph{the same }(continuously
traded) FGP, starting from $X_{1,0}=X_{2,0}$. The portfolios differ
only in their rebalancing interval. Their values are approximated
with the master equation, as if they were continuously-traded. This
approximation has been shown empirically to be accurate for diversity-
and entropy-weighted FGP approximations that are rebalanced merely
once a month \cite[Chapter 6]{Book:Fernholz:SPT:2002}. $X_{1}$ and
$X_{2}$ have the same generating functions, $H_{1}=H_{2}$, so under
this approximation their values differ only through their variance-capture
terms, $h_{1}$ and $h_{2}$, which differ only because rebalancing
occurs at different intervals, and hence different effective variance-rates
may (and do in practice) apply. The resulting approximation is
\begin{align}
\log\hat{V}_{T}^{\pi} & \approx\kappa\left(\cancel{H_{1}\left(L_{T}\right)-H_{1}\left(L_{0}\right)}+\int_{0}^{T}h_{1,s}ds-\left[\cancel{H_{2}\left(L_{T}\right)-H_{2}\left(L_{0}\right)}+\int_{0}^{T}h_{2,s}ds\right]\right)+\int_{0}^{T}h_{s}ds.\label{Eq:Stat_Arb_LS_Approx}
\end{align}
The overall portfolio $\pi$ is a constant-weight FGP, hence has $h=\gamma_{\pi}^{*}$
from (\ref{Eq:Master_General}). Therefore,
\begin{align*}
\log\hat{V}_{T}^{\pi} & \approx\int_{0}^{T}\left[\kappa\left(h_{1,s}-h_{2,s}+\frac{a_{11}-a_{22}}{2}\right)-\frac{\kappa^{2}}{2}a_{(1,-1)(1,-1)}\right]ds.
\end{align*}
This has the quadratic form of (\ref{Eq:Quad_Form_of_FGP_Stat_Arb}).
To estimate the parameters, we approximate the stochastic quantities
as constants, and plug in the values of their sample estimators. We
can identify
\begin{align*}
A & =h_{1}-h_{2}+\frac{1}{2}(a_{11}-a_{22}),\\
B & =\frac{1}{2}a_{(1,-1)(1,-1)}.
\end{align*}

In \cite{Art:Fernholz&Maguire:Stat_Arb} the FGP chosen to give the
value processes $X_{1}$ and $X_{2}$ was the equal-weight portfolio
on large cap US equities, specifically those in the S\&P500 and/or
Russell 1000 in 2005. The annualized sample averages for that year
were $a_{11}=.0683$, $a_{22}=.0423$, $a_{(1,-1)(1,-1)}=1.69\times10^{-7}$,
$h_{1}=0.0341$, $h_{2}=0.0211$.%
\footnote{These are the numbers from \cite{Art:Fernholz&Maguire:Stat_Arb},
after transforming standard deviations to variances and annualizing
all numbers:\textbf{ }$a_{11}=250*0.000273$, $a_{22}=250*0.000169$,
$a_{(1,-1)(1,-1)}=250*\left(\frac{0.0026}{100}\right)^{2}$, $h_{1}=0.0341$,
$h_{2}=0.0211$.\textbf{ } %
} These result in $A=0.0260$, $B=8.45\times10^{-8}$, $\check{\kappa}=\frac{A}{2B}=1.54\times10^{5}$,
and $\check{\gamma}_{\pi}=2.00\times10^{3}$ (base $e$).

While $\check{\kappa}$ is not of the order of magnitude usually seen
in portfolio construction, it must be remembered that it is not a
weight for investment into equities, but rather into a long-short
combination of two very diverse portfolios that are \emph{very similar
}nearly always. At the beginning of each day, the long and short portfolios
are equal, so the net position in each equity starts at 0. Each FGP
is an equal-weight portfolio in about 1000 equities. If we approximate
each initial weight as $10^{-3}$ and use a leverage factor of $\kappa=1.5\times10^{5}$,
then an isolated $1\%$ intraday price movement of a particular equity
induces a change in net weight of $1.5$ in that equity. While this
is still an unrealistically leveraged portfolio, it is much closer
to a reasonable order of magnitude considering that it would be offset
by similarly sized positions of opposite sign.

Despite the above remark, the amount of leverage involved in the $\check{\kappa}$
portfolio is prohibitive due to the realities of equity markets that
lie outside of the framework of this paper, such as price jumps, margin
requirements, transaction costs, short-selling fees, liquidity constraints,
etc. It is these factors then that become the limiting ones for the
level of leverage to use in seeking profitability from a statistical
arbitrage portfolio of this type. A more plausible level of leverage
of $\kappa=1\times10^{3}$ results in $\gamma_{\pi}=26$, still orders
of magnitude outside the realm of documented performance.

\subsection{Quadratic generating functions}

This section again considers a market whose log asset prices have
a variance rate varying with the sampling interval (e.g. Figure \ref{Fig:Variance}).
Taylor expanding the generating function term in the master equation
(\ref{Eq:Master_General}) yields
\begin{align}
\log\left(\frac{V_{T}^{\pi}}{V_{T}^{\rho}}\right) & =\left(\Delta_{T}L^{\rho}\right)^{\prime}\nabla H(L_{0}^{\rho})+\frac{1}{2}\left(\Delta_{T}L^{\rho}\right)^{\prime}D^{2}H(L_{0}^{\rho})\Delta_{T}L^{\rho}+R_{T}+\int_{0}^{T}h_{t}dt,\label{Eq:FGP_Taylor_Exp}\\
\mbox{where }\Delta_{t}L: & =L_{t}-L_{0},\nonumber
\end{align}
and $R_{T}$ is the remainder term. For a given path $\omega$, there
is a sufficiently short time horizon such that an FGP generated by
an analytic $H$ behaves nearly as if it were generated by a quadratic
$H$:
\begin{align}
H(y) & =-\frac{1}{2}\left(y-l\right)^{\prime}c\left(y-l\right)+p^{\prime}y,\qquad y,l,p\in\R^{n},\; c\in\R^{n\times n},\label{Eq:Quadratic_GF}\\
\imply\nabla H(y) & =-c\left(y-l\right)+p,\nonumber \\
\imply D^{2}H(y) & =-c.\nonumber
\end{align}
Since statistical arbitrage portfolios are generally rebalanced quite
frequently (intradaily), the above motivates consideration of FGPs
having quadratic $H$ for application in statistical arbitrage. We
assume that the investor has no information, or at least does not
wish to speculate, on the drifts of $L^{\rho}$. If this is the case,
then it makes no sense for him to take on unnecessary exposure to
the $\left(\Delta L_{T}^{\rho}\right)^{\prime}\nabla H(L_{0}^{\rho})$
term in (\ref{Eq:FGP_Taylor_Exp}). This term can be eliminated by
selecting an $H$ satisfying $\nabla H(L_{0}^{\rho})=0$.

To accomplish this initial hedge, take $p=0$ and $l=L_{0}^{\rho}$
in (\ref{Eq:Quadratic_GF}). Then
\begin{align*}
\log\left(\frac{V_{T}^{\pi}}{V_{T}^{\rho}}\right) & =-\frac{1}{2}\left(\Delta_{T}L_{i}^{\rho}\right)^{\prime}c\left(\Delta_{T}L_{i}^{\rho}\right)+\frac{1}{2}\sum_{i,j}c_{ij}\int_{0}^{T}a_{ij,s}^{\rho}ds+\int_{0}^{T}\left(\gamma_{\pi,s}^{*}-\lambda\gamma_{\rho,s}^{*}\right)ds,\\
\pi & =\lambda\rho-c\Delta_{T}L^{\rho}.
\end{align*}
For simplicity in illustrating the idea, we restrict the investment
to one risky asset (possibly a wealth process of a more general portfolio,
as in the previous section) and one locally risk-free asset, i.e.
a money market account, which will be the numéraire. The procedure
is readily generalizable to a bigger market, with the cost being solving
and optimizing vector instead of scalar equations. In this setting
$\gamma_{\rho}^{*}=0$, $L^{\rho}=\hat{L}$, and $\pi=\lambda\rho-c\Delta\hat{L}$.
Since the money market discounted with itself has value one for all
time, then $H$ only need be prescribed on the risky discounted asset's
log price, $\hat{L}$, and hence $c$ is a scalar. The log wealth
is
\begin{align}
\log\hat{V}_{T}^{\pi} & =\frac{1}{2}c\left(\int_{0}^{T}a_{s}ds-\left(\Delta_{T}\hat{L}\right)^{2}\right)+\int_{0}^{T}\gamma_{\pi,s}^{*}ds,\nonumber \\
 & =\frac{1}{2}\left[\int_{0}^{T}\left(ca_{s}+a_{s}\left[c\Delta_{s}\hat{L}\right]_{i}-c^{2}\left[\Delta_{s}\hat{L}\right]^{2}a_{s}\right)ds-c\left(\Delta_{T}\hat{L}\right)^{2}\right].\label{Eq:Quad_FGP_Int}
\end{align}
To proceed, we take an expectation, assume Brownian integrals are
martingales, and approximate some time-dependent parameters as constants.
This simplifies the model, allowing for easy fitting to data:
\begin{align*}
E[a_{t}] & \approx:a\\
E\Delta_{s}\hat{L} & =E\left[\int_{0}^{s}\hat{\gamma}_{u}du\right]\approx:s\hat{\gamma},\\
E\left[\left(\Delta_{s}\hat{L}\right)^{2}\right] & =:A_{s}.
\end{align*}
Note: $a$ is best interpreted as $A_{t}/t$ at the $t$ at which
trading actually occurs. Using the above in (\ref{Eq:Quad_FGP_Int})
yields
\begin{align}
E\log\left(\hat{V}_{T}^{\pi}\right) & \approx\frac{1}{2}\left[acT+ac\hat{\gamma}\int_{0}^{T}sds-ac^{2}\int_{0}^{T}A_{s}ds-cA_{T}\right],\nonumber \\
 & =\frac{Ta}{2}\left(c\left[1-\frac{A_{T}}{Ta}+\frac{\hat{\gamma}T}{2}\right]-c^{2}v(T)\right),\label{Eq:Stat_Arb_Ret_w_Drift}\\
\mbox{where }v(T): & =\frac{1}{T}\int_{0}^{T}A_{s}ds.\nonumber
\end{align}
By assumption, $A_{t}/at$ is not identically $1$. If its deviation
from $1$ is not substantially greater than $|\hat{\gamma}t/2|$ for
some $t>0$, then our exposure to $\hat{\gamma}$ results in large
risk for little gain. In any case, we are ignorant of the drift, so
we drop it and are left with
\begin{align*}
\frac{1}{T}E\log\left(\hat{V}_{T}^{\pi}\right) & \approx\frac{a}{2}\left(c\left[1-\frac{A_{T}}{Ta}\right]-c^{2}v(T)\right).
\end{align*}
All that is needed is the variogram for $\hat{L}$ , from which the
other quantities are easily derived. The $c$ that maximizes the expected
log growth by horizon $T$ is
\begin{align*}
\check{c}_{T} & =\frac{1}{2v(T)}\left(1-\frac{A_{T}}{Ta}\right).
\end{align*}
This yields a maximal expected log-growth rate of
\begin{align}
\frac{1}{T}E\log\left(\hat{V}_{T}^{\pi}\right) & =\frac{a}{8v(T)}\left(1-\frac{A_{T}}{Ta}\right)^{2}.\label{Eq:Stat_Arb_QGF_Rate_Ret}
\end{align}
Empirically, the quantity $A_{T}/T$ tends to a constant for large
$T$, thus $v(T)=\cO(T)$ for $T\to\infty$. This means that the growth
rate tends to $0$ as $T\to\infty$. The question then arises of what
the optimal period $\check{T}$ is for restarting this strategy. This
can be obtained by maximizing (\ref{Eq:Stat_Arb_QGF_Rate_Ret}) as
a function of $T$.
\begin{conjecture}
Rather than restarting the portfolio after a given time period, it
may be better to solve an optimal control problem, restarting when
the price of the risky asset wanders sufficiently far from its origin.
\end{conjecture}
A positive feature of both the methodologies of Section \ref{Sec-Stat_Arb}
is that they are entirely data-driven, and depend only on variance
measurements, which can be estimated in practice with high precision.
Of course if the data suggests a parametric model, then the added
structure could be additionally exploited.

\noindent \begin{center}
\begin{figure}
\noindent \centering{}\includegraphics{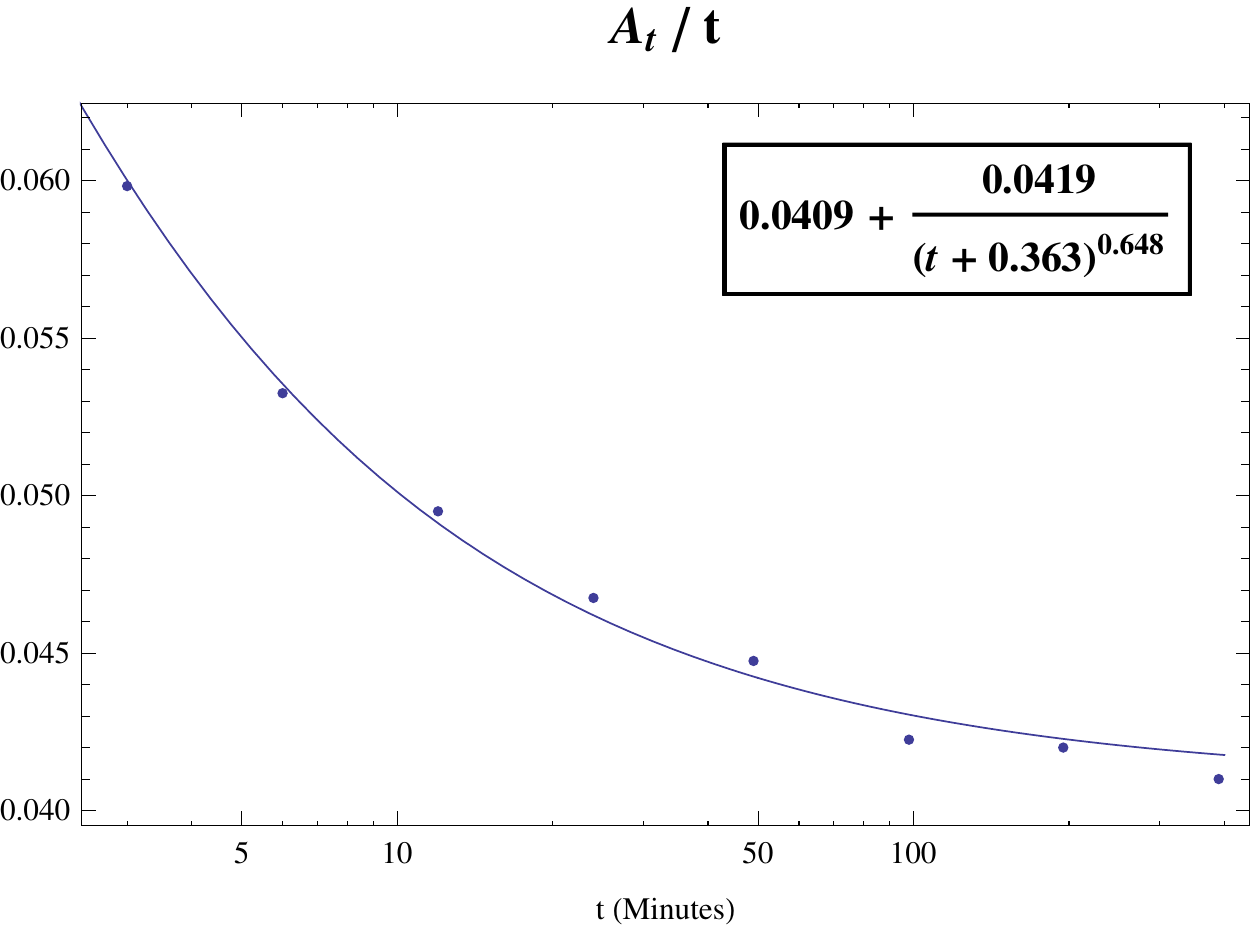}\caption{\label{Fig:Variance}Our fit to the annualized variogram of US-large
cap stocks from \cite[Figure 1]{Art:Fernholz&Maguire:Stat_Arb}. }
\end{figure}

\par\end{center}
\begin{example}
\label{Ex:Arb_Quad_GF}We apply a quadratic generating function to
the data from \cite{Art:Fernholz&Maguire:Stat_Arb}, i.e. of an equal-weight
portfolio on large cap US equities in 2005. The variogram (Figure
\ref{Fig:Variance}) is fitted to
\begin{align*}
\frac{A_{t}}{t} & =C+\frac{U}{\left(t+B\right)^{k}}.
\end{align*}
The form was chosen for fitting fairly well and also because $A$
has a closed-form antiderivative, yielding
\begin{align*}
v(T) & \equiv\frac{1}{T}\int_{0}^{T}A_{t}dt=\frac{CT}{2}+Uk\frac{\left(B+T\right)^{1-k}\left(\left[1-k\right]T-B\right)+B^{2-k}}{T\left(2-k\right)\left(1-k\right)}.
\end{align*}
Assuming, as in Section \ref{Sub-Long-Short_Stat_Arb}, that our
(fastest) rebalancing frequency is $1.5$ minutes, then effectively
$a=A_{1.5}/1.5=.0683$ (annualized). From (\ref{Eq:Stat_Arb_QGF_Rate_Ret})
we numerically obtain $\check{T}=7.13$ minutes, leading to $\check{c}(\check{T})=5.8\times10^{4}$,
and a maximal rate of log return of $244$ (base $e$, per year).
This is an order of magnitude lower than $\hat{\gamma}_{\pi}$ obtained
from the same data via the methodology in Section \ref{Sub-Long-Short_Stat_Arb}.
This may be explained by the approximation made in (\ref{Eq:Stat_Arb_LS_Approx})
breaking down under substantial leverage. Another factor may be that
our dropping the drift term $\hat{\gamma}T/2$ of (\ref{Eq:Stat_Arb_Ret_w_Drift})
underestimates the return of the quadratic FGP: Since $A_{T}/T<a$,
then conditioned on a price movement, a large drift in the opposite
direction is typical, which the portfolio implicitly bets on.
\begin{figure}
\noindent \centering{}\includegraphics{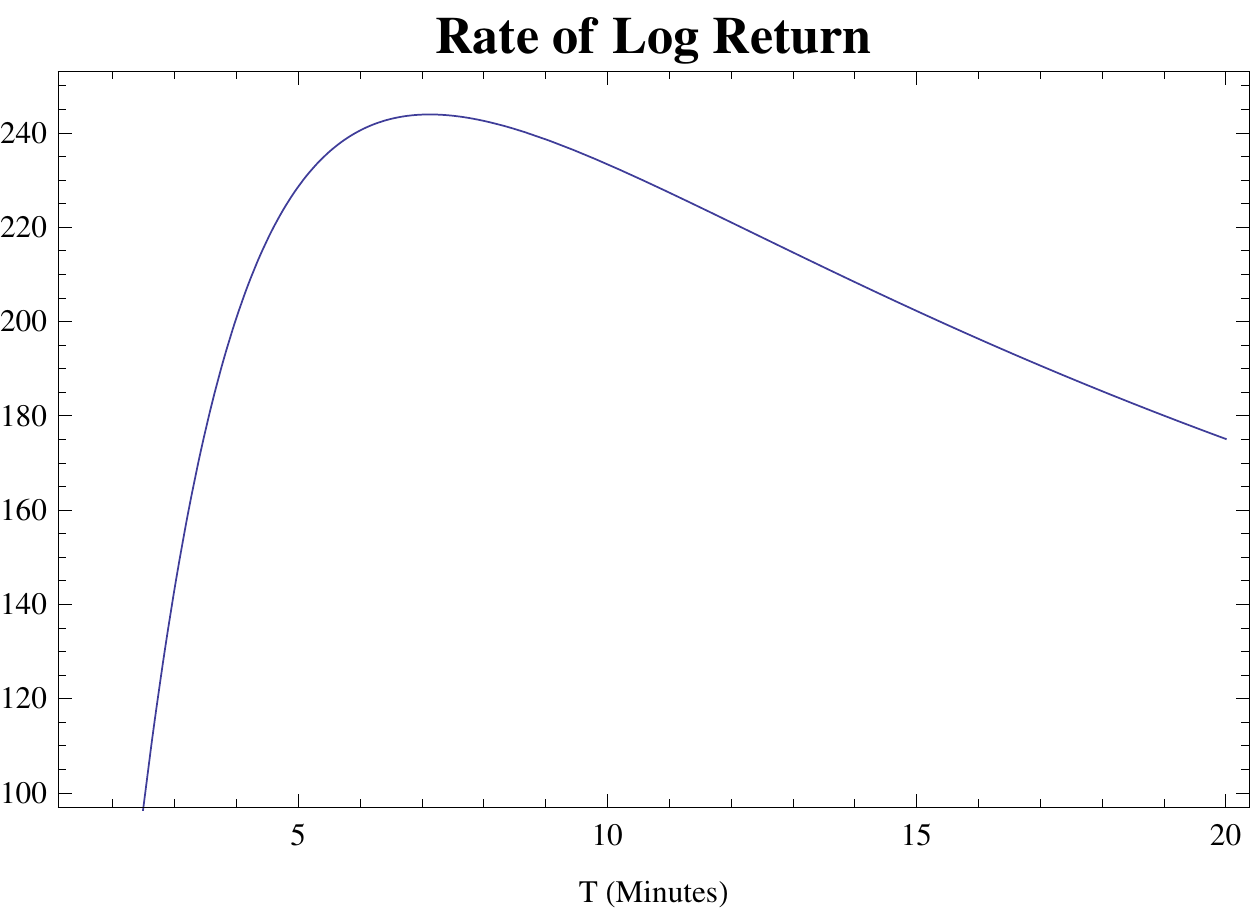}\caption{\label{Fig:Log_Ret}The optimal annualized rate of log-return as a
function of horizon in Example \ref{Ex:Arb_Quad_GF}. }
\end{figure}

\end{example}

\section{\label{Sect-Immunization}Portfolio immunization}

Suppose that we have selected a generating function that is appealing,
except that it produces a generated portfolio that is exposed, relative
to the numéraire, to risk factors that we would rather remain unexposed
to. For example, we may wish to avoid taking on numéraire risk by
maintaining zero excess exposure to it.

One way to remove unwanted risk exposure is to modify the initial
generating function, only in so far as to make it invariant to changes
in the argument along the direction of given risk factors. To be more
concrete, suppose that $H$ is the initial generating function, and
that $\beta^{1},\ldots,\beta^{K}$ are each continuous-path finite
variation processes in $\R^{n}$ satisfying
\begin{align*}
\left(\beta_{t}^{k}\right)^{\prime}\beta_{t}^{j} & =\delta_{kj},\quad1\le k,j\le K,\quad\forall t\ge0.
\end{align*}
The orthonormal set of random vectors $\{\beta_{t}^{1},\ldots,\beta_{t}^{K}\}$
spans the subspace in $\R^{n}$ that we would like to immunize the
generated portfolio's performance to at time $t$. That is, we would
like to find a generating function $\tilde{H}$, similar to $H$,
except also obeying
\begin{align*}
\left(\beta_{t}^{k}\right)^{\prime}\nabla\tilde{H}_{t} & =0,\quad\forall t\ge0.
\end{align*}
For this to hold, $\tilde{H}$ will need to be stochastic, via taking
$\beta:=\{\beta_{i}^{k}\}_{1\le i\le n}^{1\le k\le K}$ as a second
argument. Perhaps the most natural way to modify $H$ in order to
achieve this is to project $\nabla H$ onto the complement of the
span of $\{\beta^{1},\ldots,\beta^{K}\}$ and allow that to determine
the new function $\tilde{H}$. To that end, let $P^{\perp}(y,b)$
be the projection operator that projects $y$ onto the orthogonal
complement of the span of vectors $\{b^{k}\}^{1\le k\le K}$. It is
simplest to specify $P^{\perp}$ in the case where $\{b^{k}\}$ are
orthonormal:
\begin{align*}
P^{\perp}(y,b) & =y-\sum_{k=1}^{K}\left(y^{\prime}b^{k}\right)b^{k},\quad\mbox{for }\left(b^{k}\right)^{\prime}b^{j}=\delta_{kj},\;1\le k,j\le K.
\end{align*}

\begin{prop}
\label{Prop:Master_Immunization}Let $\{\beta^{1},\ldots,\beta^{K}\}$
be $K\le n$ finite variation processes in $\R^{n}$ that are mutually
orthonormal at all times. The generating function $H$ and generated
portfolio
\begin{align*}
\pi & =\lambda\rho+P^{\perp}(\nabla H(P^{\perp}(L^{\rho},\beta)),\beta),\qquad\mbox{where }\lambda=1-\I^{\prime}P^{\perp}(\nabla H(P^{\perp}(L^{\rho},\beta)),\beta),
\end{align*}
satisfy the following master equation:
\begin{align*}
\log\left(\frac{V_{T}^{\pi}}{V_{T}^{\rho}}\right) & =H(P^{\perp}(L_{T}^{\rho},\beta_{T}))-H(P^{\perp}(L_{0}^{\rho},\beta_{0}))+\int_{0}^{T}h_{t}dt-\sum_{k=1}^{K}\int_{0}^{T}\left[\nabla_{b^{k}}H(P^{\perp}(L_{t}^{\rho},\beta_{t}))\right]^{\prime}d\beta_{t}^{k},
\end{align*}
where
\begin{align*}
h & =\gamma_{\pi}^{*}-\lambda\gamma_{\rho}^{*}-\frac{1}{2}\biggl(\sum_{i,j=1}^{n}a_{ij}^{\rho}D_{ij}^{2}H(P^{\perp}(L^{\rho},\beta)),\beta)-2\sum_{k=1}^{K}\left(b^{k}\right)^{\prime}a^{\rho}D^{2}H(P^{\perp}(L^{\rho},\beta))b^{k}\\
 & \quad+\sum_{k=1}^{K}\sum_{k^{\prime}=1}^{K}\left(b^{k^{\prime}}\right)^{\prime}a^{\rho}b^{k}\left(b^{k^{\prime}}\right)^{\prime}D^{2}H(P^{\perp}(L^{\rho},\beta))b^{k}\biggr),\\
\nabla_{b^{k}}H(P^{\perp}(y,b)) & =-\left[\left(b^{k}\right)^{\prime}\nabla H(P^{\perp}(y,b))\right]y_{i}-\left[\left(b^{k}\right)^{\prime}y\right]D_{i}H(P^{\perp}(y,b)).
\end{align*}
\end{prop}
\begin{proof}
Define
\begin{align*}
\tilde{H}:\R^{n}\times\R^{K\times n} & \to\R,\\
\tilde{H}(y,b) & =H(P^{\perp}(y,b)),\\
\mbox{where }P^{\perp}(y,b) & =y-\sum_{k=1}^{K}\left(y^{\prime}b^{k}\right)b^{k}.
\end{align*}
The result is then obtained from a direct application of Theorem \ref{Thm:Master_Eq_Stoch_GF}
to $\tilde{H}$. The relevant derivatives are
\begin{align*}
\nabla_{y}\tilde{H}(y,b) & =\nabla H(P^{\perp}(y,b))-\sum_{k=1}^{K}\left[\left(b^{k}\right)^{\prime}\nabla H(P^{\perp}(y,b))\right]b^{k},\\
 & =P^{\perp}(\nabla H(P^{\perp}(y,b)),b);\\
\frac{\partial^{2}}{\partial y_{i}\partial y_{j}}\tilde{H}(y,b) & =D_{ij}^{2}H(P^{\perp}(y,b))-\sum_{k=1}^{K}b_{j}^{k}\left[\left(b^{k}\right)^{\prime}D^{2}H(P^{\perp}(y,b))\right]_{i}-\sum_{k=1}^{K}b_{i}^{k}\left[\left(b^{k}\right)^{\prime}D^{2}H(P^{\perp}(y,b))\right]_{j}\\
 & \quad+\sum_{k=1}^{K}\sum_{k^{\prime}=1}^{K}b_{i}^{k}b_{j}^{k^{\prime}}\left(b^{k^{\prime}}\right)^{\prime}D^{2}H(P^{\perp}(y,b))b^{k};
\end{align*}
\begin{align*}
\frac{\partial}{\partial b_{i}^{k}}\tilde{H}(y,b) & =-\left(\left(b^{k}\right)^{\prime}\nabla H(P^{\perp}(y,b))\right)y_{i}-\left(\left(b^{k}\right)^{\prime}y\right)D_{i}H(P^{\perp}(y,b)).\tag*{\qedhere}
\end{align*}

\end{proof}
The characterization of the performance of the immunized FGP
given by Proposition \ref{Prop:Master_Immunization} is not so pretty,
but the idea of what has changed from the non-immunized FGP is straightforward.
The relative wealth process $\log\left(V^{\pi}/V^{\rho}\right)$ of
the generated portfolio of Proposition \ref{Prop:Master_Immunization}
is locally not exposed to changes in $L^{\rho}$ along the linear
span of $\{\beta^{1},\ldots,\beta^{K}\}$. This can be seen from
\begin{align*}
\left(\beta^{k}\right)^{\prime}\nabla_{y}H(P^{\perp}(y,b))\mid_{(L^{\rho},\beta)} & =\left(\beta^{k}\right)^{\prime}\left(\pi-\lambda\rho\right),\\
 & =\left(\beta^{k}\right)^{\prime}P^{\perp}(\nabla H(P^{\perp}(L^{\rho},\beta)),\beta),\\
 & =0,\quad1\le k\le K.
\end{align*}

\begin{example}[Numéraire exposure]
Consider the case where immunization is desired with respect to relative
exposure to the numéraire. The appropriate $\beta$ to use to hedge
against excess numéraire exposure is $1$ less than the ``CAPM $\beta$''
(see e.g. \cite{Book:Sharpe_Port_Theory_and_Cap_Markets}). The instantaneous
version of this parameter is
\begin{align*}
\tilde{\beta}_{i,t} & =\frac{\frac{d}{dt}\left\langle \log V^{\rho},L_{i}^{\rho}\right\rangle _{t}}{\frac{d}{dt}\left\langle \log V^{\rho}\right\rangle _{t}}=\frac{[a_{t}r_{t}]_{i}-a_{\rho\rho,t}}{a_{\rho\rho,t}}=\frac{[a_{t}\rho_{t}]_{i}}{a_{\rho\rho,t}}-1,\\
 & =\beta_{i,t}^{\left(\mbox{CAPM}\right)}-1.
\end{align*}
Although theoretically this instantaneous $\beta$ may not be a continuous-path
finite variation process, in practice the instantaneous $\beta$ is
not observable, and $\beta$ is typically estimated by time-averaging
over some historical time window. The practical \emph{and }theoretical
result of such a time-averaging procedure is a continuous-path finite
variation process. For example, the estimator might have the theoretical
form
\begin{align*}
\tilde{\beta}_{i,t} & =\frac{\frac{1}{\Delta t}\int_{t-\Delta t}^{t}[a_{s}\rho_{s}]_{i}ds}{\frac{1}{\Delta t}\int_{t-\Delta t}^{t}a_{\rho\rho,s}ds}-1,
\end{align*}
for some $\Delta t>0$. In practice the integrals are approximated
by sums of discretely sampled values.
\end{example}

\begin{example}[Price level]
Another possibly desirable immunization is to hedge out any exposure
to a rise or fall in the overall price level. This can be done by
choosing the constant vector $\beta=n{}^{-1/2}\I$.
\end{example}

\section{\noindent \label{Sec-Mirror_Ports}Mirror portfolios}

In this section we use generating functions to elaborate some of the
properties of \emph{mirror portfolios}, introduced in \cite{Art:FernKaratzKard:DiversityAndRelArb:2005}.
In that paper, mirror portfolios were used to construct arbitrages
over arbitrarily short time horizons in markets that are both diverse
and uniformly elliptic. A passive portfolio that is short any asset
is typically inadmissible, due to each asset's price usually being
unbounded from above. Hence, flipping the sign of the shares invested
in assets (while adjusting the weight in the numéraire so that weights
sum to one) does not in general produce a suitable notion of a reflected
portfolio. Mirror portfolios accomplish that task.

\noindent
\begin{defn}
\noindent If $\pi$ and $\rho$ are portfolios, then the portfolio
\begin{align*}
\tilde{\pi}^{[q],\rho}:= & q\pi+(1-q)\rho,\quad q\in\R,
\end{align*}
is called the \emph{$q$-mirror of $\pi$ with respect to $\rho$}.
When $\rho$ is fully invested in the money market, then $\tilde{\pi}^{[q],\rho}=:\tilde{\pi}^{[q]}$,
abbreviated the\emph{ $q$-mirror of $\pi$}. For $q=-1$, $\tilde{\pi}^{[-1],\rho}=:\tilde{\pi}^{\rho}$,
called simply the \emph{mirror of} $\pi$ \emph{with respect to $\rho$}.
\end{defn}
\noindent For portfolios $\pi$ and $\rho$, the $q$-mirror of $\pi$
with respect to $\rho$ satisfies (\ref{Eq:Integrability_Cond_to_be_a_Port}),
so is also a portfolio. As an example, if $X_{1}$ is the money market,
then the portfolio $e_{i}:=(0,\ldots,0,1,0,\ldots)$ has the mirror
$\tilde{e}_{i}=(2,0,\ldots,0,-1,0,\ldots)$.
\begin{prop}
\noindent The $q$-mirror of $\pi$ with respect to $\rho$ is functionally
generated from the market $(X_{1},X_{2}):=(V^{\rho},V^{\pi})$ by
the generating function $H(y_{1},y_{2}):=\left(1-q\right)y{}_{1}+qy_{2}$,
and thus satisfies
\begin{align*}
\log V^{\tilde{\pi}^{[q],\rho}} & =(1-q)\log V^{\rho}+q\log V^{\pi}+\int h_{s}ds,\\
\log\left(\frac{V^{\tilde{\pi}^{[q],\rho}}}{V^{\rho}}\right) & =q\log\left(\frac{V^{\pi}}{V^{\rho}}\right)+\int h_{s}ds,\\
\mbox{where }h=\gamma_{\tilde{\pi}^{[q],\rho}}^{*} & =\frac{1}{2}\left((1-q)a_{11}+qa_{22}-\left[\left(1-q\right)^{2}a_{11}+2q(1-q)a_{12}+q^{2}a_{22}\right]\right).
\end{align*}
If, additionally, $\rho$ is the money market, then
\begin{align*}
\log\hat{V}^{\tilde{\pi}^{[q]}} & =q\log\hat{V}^{\pi}+\frac{q(1-q)}{2}\left\langle \log V^{\pi}\right\rangle .
\end{align*}
If, additionally, $q=-1$, then
\begin{align}
\log\hat{V}^{\tilde{\pi}} & =-\log\hat{V}^{\pi}-\left\langle \log V^{\pi}\right\rangle ,\label{Eq:Mirror_Master}
\end{align}
\end{prop}
\begin{proof}
\noindent $H$ is translation equivariant and $D^{2}H=0$, so we may
apply Proposition \ref{Prop:Master_Eq_for_PH_GF} to obtain the first
result. The others are easy consequences of plugging in $a_{11}=a_{12}=0$
when $\rho$ is the money market.
\end{proof}
The following corollary shows that under typical market conditions
a given portfolio $\pi$ or its mirror $\tilde{\pi}$ or both will
lose all wealth, asymptotically.
\begin{cor}
\noindent \label{Cor:Mirror_Decay}Suppose that both of the following
hold:
\begin{enumerate}
\item \noindent
\begin{align}
\liminf_{t\to\infty}\frac{1}{t}\left\langle \log V^{\pi}\right\rangle _{t} & >0,\quad\mbox{a.s.}\label{Eq:Non-negl_Asymp_QV}
\end{align}

\item
\begin{align}
\lim_{t\to\infty}\frac{\log\log t}{t^{2}}\left\langle \log V^{\pi}\right\rangle _{t} & =0,\quad\mbox{a.s.}\label{Eq:Iterated_Log_Cond}
\end{align}

\end{enumerate}
Then
\begin{align*}
P\left(\left\{ \lim_{t\to\infty}\hat{V}_{t}^{\pi}=0\right\} \bigcup\left\{ \lim_{t\to\infty}\hat{V}_{t}^{\tilde{\pi}}=0\right\} \right) & =1.
\end{align*}
\end{cor}
\begin{proof}
\noindent Under (\ref{Eq:Iterated_Log_Cond}) the law of the iterated
logarithm \cite[p. 112]{Book:KS:Brownian:1991} implies that
\begin{align}
\lim_{t\to\infty}\frac{1}{t}\left(\log\hat{V}_{t}^{\pi}-\int_{0}^{t}\gamma_{\pi,s}ds\right) & =0,\quad\mbox{a.s.,}\label{Eq:Long_Term_Growth}
\end{align}
since the process inside the parentheses is a continuous local martingale.
From this (\ref{Eq:Mirror_Master}) yields
\begin{align*}
\lim_{t\to\infty}\frac{1}{t}\left(\log\hat{V}_{t}^{\tilde{\pi}}+\left\langle \log V^{\pi}\right\rangle _{t}+\int_{0}^{t}\gamma_{\pi,s}ds\right) & =0,\quad\mbox{a.s.,}\\
\imply\lim_{t\to\infty}\frac{1}{t}\left(\log\hat{V}_{t}^{\pi}+\log\hat{V}_{t}^{\tilde{\pi}}+\left\langle \log V^{\pi}\right\rangle _{t}\right) & =0,\quad\mbox{a.s.,}
\end{align*}
where the second line follows from adding (\ref{Eq:Long_Term_Growth})
to the first. Then by (\ref{Eq:Non-negl_Asymp_QV})
\begin{align*}
\limsup_{t\to\infty}\frac{1}{t}\left(\log\hat{V}_{t}^{\pi}+\log\hat{V}_{t}^{\tilde{\pi}}\right) & <0,\quad\mbox{a.s.,}\\
\imply P\left(\left\{ \limsup_{t\to\infty}\frac{1}{t}\log\hat{V}_{t}^{\pi}<0\right\} \bigcup\left\{ \limsup_{t\to\infty}\frac{1}{t}\log\hat{V}_{t}^{\tilde{\pi}}<0\right\} \right) & =1.\tag*{\qedhere}
\end{align*}

\end{proof}
\noindent Equation (\ref{Eq:Mirror_Master}) shows that at least one
and possibly both of $\log\hat{V}^{\pi}$ and $\log\hat{V}^{\tilde{\pi}}$
have negative drift at any time when $\left\langle \log V^{\pi}\right\rangle $
is increasing. The preceding corollary shows that a portfolio, its
mirror, or possibly both, lose all wealth relative to the money market
asymptotically, assuming that the asymptotic local variance rate does
not approach 0. A portfolio whose wealth tends to $0$ asymptotically
would typically be considered a poor long-term investment. In this
sense, ``mirroring'' a poor investment may still be a poor investment.
A concrete example is a market with a risk-free rate of $0$, and
one risky asset whose price is a geometric Brownian motion with $\gamma=-\frac{1}{2}\sigma^{2}$.
Then full investment in the risky asset loses all wealth asymptotically,
as does its mirror, which also has drift $\tilde{\gamma}=-\frac{1}{2}\sigma^{2}$
by (\ref{Eq:Mirror_Master}).

\section{\label{Sect:Conclusion}Concluding remarks}

The key analytical benefit of portfolios that are functionally generated
is the representation of their return relative to a numéraire via
a pathwise master equation free of stochastic integrals. The generalizations
of FGPs presented here expand the class of portfolio-numéraire pairs
that may be analyzed in this way. The dynamism of FGPs is enhanced
by the freedom to incorporate processes having continuous, finite-variation
paths as auxiliary arguments to generating functions. This allows
FGPs to be sensitive to changing market conditions beyond the price
changes of the assets. The main applications that we have shown are
(1) direct, intuitive comparison of the performance of FGPs, useful
for scenario analysis (Section \ref{Sub:Scenario}),  (2) statistical
arbitrage based purely on variance data, (3) portfolio immunization,
and (4) mirror portfolios analysis.

It is a shortcoming of this work that transaction costs are ignored
throughout. They are especially important to the performance of the
statistical arbitrage portfolios examined in Section \ref{Sec-Stat_Arb}.
The inclusion of transaction costs in a tractable way for FGPs in
$n$-asset markets is a topic of ongoing research. Due to its complexity,
it warrants a separate paper that the author hopes will be forthcoming
in the future.

\bibliographystyle{hplain}

%\bibliography{MathFinBib}
\end{document}